\documentclass{endm}
\usepackage{endmmacro}
\usepackage{graphicx}

\usepackage{dsfont}

\def\N{{\mathbb N}}

\newcommand{\Bn}{{\mathcal{B}_n}}

\newcommand{\n}{$n$}
\newcommand{\cutn}{\ensuremath{\operatorname{CUT}(n)}}

\newcommand{\ind}{\mathds{1}}

\begin{document}

% DO NOT REMOVE: Creates space for Elsevier logo, ScienceDirect logo
% and ENDM logo
\begin{verbatim}\end{verbatim}\vspace{2.5cm}

\begin{frontmatter}

\title{Cut polytope has vertices on a line}

\author{Nevena Mari\'c %\thanksref{ALL}\thanksref{maric.math.umsl.edu}}
}
\address{Department of Mathematics and Computer Science\\ University of Missouri - St. Louis\\ St. Louis, USA}

%\author{My Co-author\thanksref{coemail}}
%\address{My Co-author's Department\\ My Co-author's University\\
%   My Co-author's City, My Co-author's Country}
%    \thanks[ALL]{Thanks
%   to everyone who should be thanked} 
   \thanks[myemail]{Email:
   \href{mailto:maric@math.umsl.edu} {\texttt{\normalshape
   maric@math.umsl.edu}}} 
%   \thanks[coemail]{Email:
%   \href{mailto:couserid@codept.coinst.coedu} {\texttt{\normalshape
%   couserid@codept.coinst.coedu}}}

\begin{abstract}

The cut polytope ${\rm CUT}(n)$ is the convex hull of
  the cut vectors in a complete graph with vertex set $\{1,\ldots,n\}$. It is well known in the area of combinatorial optimization and recently has also been studied in a direct relation with admissible correlations of symmetric Bernoulli random variables. That probabilistic interpretation is a starting point of this work in conjunction with a natural binary encoding of the CUT($n$). We show that for any $n$, with appropriate scaling, all encoded vertices of the polytope ${\mathbf 1}$-CUT($n$) are approximately on the line $y= x-1/2$.
  
%  The polytope  $R(\Bn)$ was characterized by identifying
%  its vertices and  it was shown also that it is affinely isomorphic to
%  the well-known cut polytope ${\rm CUT}(n)$ .
%  The isomorphism was obtained explicitly as
%  $R(\Bn)= {\mathbf{1}}-2~{\rm CUT}(n)$.   
%The cut polytope CUT($n$) is  well known in the area of combinatorial optimization \cite{dezalaurent1997} and its approximation by $\mathcal{E}_{n}$ has most recently been studied in \cite{tropp2018}.
%

\end{abstract}

\begin{keyword}
Cut polytope, Bernoulli correlations.
\end{keyword}

\end{frontmatter}

\section{Introduction}\label{intro}

The cut polytope ${\rm CUT}(n)$ is  the convex hull of
   cut vectors in a complete graph with vertex set $\{1,\ldots,n\}$. It is  well known in the area of combinatorial optimization as it can be used to formulate the max-cut problem, which has many applications in various fields, like statistical physics, in relation to spin glasses~\cite{dezalaurent1997}. It has also been studied in relation with correlations of binary random variables.  A symmetric Bernoulli random variable is such that takes values 0 and 1 with equal probabilities. The space of all $n$-variate symmetric Bernoulli r.v. is denoted by $\Bn$ and its correlation space $R(\Bn)$. It s well known that every correlation matrix belongs to $\mathcal{E}_{n}$, the set of  symmetric positive semi-definite matrices with all diagonal elements equal to 1.  For Gaussian marginals, the entirety of $\mathcal{E}_{n}$ can be realized, but surprisingly enough, for other common distributions very little is known \cite{huberm2015}. 
%In spite of wide usage of the correlation coefficient,  its bounds in a multivariate non-Gaussian setting has been mostly unexplored. 
%A notable partially explored example is that of {\em copulas} \cite{devroye2015copulas} for $n  \leq 9$.
For multivariate symmetric Bernoulli the problem was recently solved  in \cite{huberm2017} where the polytope  $R(\Bn)$ was characterized by identifying
  its vertices. A relationship with the CUT($n$) was established explicitly as $R(\Bn)= {\mathbf{1}}-2~{\rm CUT}(n)$.  
 %As a corollary of this,
%a sampling algorithm for multivariate symmetric Bernoulli with given correlation was obtained. 
%  We were able to show that the space of multivariate symmetric Bernoulli correlations $R(\Bn)$ is affinely isomorphic to the  {\em cut polytope} CUT($n$),  explicitly as $R(\Bn)= {\mathbf{1}}-2~{\rm CUT}(n)$. 
A relation between CUT($n$) and its approximation by $\mathcal{E}_{n}$ has also  been studied in \cite{tropp2018}.
 In this work the established relationship of {\bf 1}-CUT($n$) with $R(\Bn)$ is our starting point. We then use a natural binary encoding to study the vertices of {\bf 1}-CUT($n$) as integers. It is shown that, with appropriate scaling, for any $n$, the encoded vertices of this polytope are approximately on the line $y=x-1/2$. Consequently, the encoded vertices of CUT($n$) are  approximately on the line $y= -x + 2^{n-1}+1/2$.

\section{Cut polytopes}\label{sec:cut}

Let $G = (V,E)$ be a graph with vertex set $V$ and edge set $E$. For $S \subseteq V$
 a {\em cut} of the graph is a partition $(S,S^C)$ of the vertices. The {\em cut-set} consists of all edges that connect a node in $S$ to a node not in $S$. 

Let $V_n=[n]= \{1,\ldots,n\}$, $E_n= \{(i,j); 1 \leq i \neq  j \leq n \}$,  and $K_n=(V_n, E_n)$ be a complete graph with the vertex set $[n]$.
\begin{definition}\label{def:cut}
For every $ S \subseteq [n]$ a vector $\delta(S) \in \{0,1\}^{E_n}$, defined as 

\begin{eqnarray*}
\delta(S)_{ij}= \left\{ \begin{array}{cc}
1, & \mbox{ if }  |S \cap \{i,j\}|=1 \\
0, & \mbox{ otherwise }
  \end{array} \right.
  \end{eqnarray*}
  for $(1 \leq i < j \leq n)$, is called a {\em cut vector} of $K_n$.
  \end{definition}

The {\em cut polytope} CUT(\n) is the convex hull of all cut vectors of $K_n$:
\[
{CUT}(n) = conv\{ \delta(S): S \subseteq [n] \}.
\]

\begin{remark}
Since every cut vector is a vertex of \cutn, there are $2^{n-1}$ vertices 
of this polytope~\cite{ziegler2000}.
\end{remark}
Each $\delta(\cdot)$ is a $0/1$-vector (every coordinate value is either 0 or 1). The convex hulls of finite sets of $0/1$-vectors are called {\em 0/1-polytopes}, out of which cut polytopes are a sub-class. An excellent lecture on 0/1-polytopes, including CUT,  is given in Ziegler~\cite{ziegler2000}. 
A thorough treatment of cut polytopes can be found in Deza and Laurent~\cite{dezalaurent1997}. 
%(\textcolor{red}{when they appeared? why they are important?})
%The cut polytopes play an important role in combinatorial optimization, as they can be used to formulate the max-cut problem, which has many applications in various fields, like statistical physics, in relation to spin glasses~\cite{dezalaurent1997}.
The starting point for us here are results from Huber and Mari\'c \cite{huberm2017} where the cut polytopes are given a new probabilistic interpretation.

\section{Cut polytopes via agreement probabilities}
 
 \begin{definition} \label{deflambda}
 For an $n$-dimensional $0/1$-vector $x$ we define its {\em concurrence vector} as
 $\lambda(x)=(\lambda(x)_{12},\lambda(x)_{13},\ldots,\lambda(x)_{1n},\lambda(x)_{23},\ldots,\lambda(x)_{2n},\ldots,\lambda(x)_{n-1\,n})$ where
 $ \lambda(x)_{ij}= \ind (x(i)=x(j))     $, for $ 1 \leq i < j \leq n$. 
 \end{definition}
 Here $\ind$ denotes the indicator function: $\ind(A) = 1$ if $A$ is true and $0$ otherwise. Applying the definition to an example $x=(0,1,1,0)$, $\lambda(x)=(0,0,1,1,0,0)$. Note that if $x$ has $n$ coordinates then $\lambda(x)$ has $n \choose 2$ coordinates. 
 
 Introduction of the concurrence vector has its motivation from the context of symmetric Bernoulli random variables \cite{huberm2015}. Let $B_1,\ldots, B_n \sim Bern(1/2)$, that is $P(B_i=1)=P(B_i=0)=1/2$, for all $i$. The random vector $(B_1,\ldots,B_n) \in \Bn$ takes values in $\{0,1\}^n$ and correlations among these variables are explicitly related to {\em concurrence} probabilities, i.e. probabilities of two variables taking the same value, $P(B_i =B_j)$, for $i \neq j$. 
 
 Let's look at  elements of $\Bn$ (the set of all $n$-variate symmetric Bernoulli dist.) that are uniformly distributed over two diagonal points of $\{0,1\}^n$, where by a diagonal we mean the set $\{x$, $ {\bf{1}}-x \}$. There are $2^{n-1}$ such distributions and they play an important role for both $\Bn$ and $R(\Bn)$. Namely it was shown in \cite{huberm2017} that  the concurrence vectors associated to those diagonal distributions are precisely vertices of the polytope {\bf 1} - CUT($n$) (obtained by replacing all coordinates $x_i$ by $1-x_i$).  
 %It is not hard to see that they can be calculated precisely in the way used in the above definition.
 
 \subsection{Binary encoding}
 
Let us look now at elements of $\{0,1\}^n$ encoded as a binary representation of numbers $\{0,1, 2, \ldots, 2^{n}-1\}$.  For instance we will identify $(0,1,1,1)$ with  a binary number $0111$, that is decimal number 7.  Note that $00111$ also represents decimal number 7, so when needed we will specify the number of bits used in representation of the specific number. The notation in that case will be $x_{[k]}$,  where $k$ is the number of bits. When it is helpful for easier reading to emphasize that the number is represented in binary, we will add $b$ in superscript, like $x^b_{[k]}$.
{\bf  More notation:} We will write two strings next to each other with vertical dots in between to denote concatenation: if $x = 001$, then $0\vdots x = 0001$. Also $\bar x$ is the complement of $x$. 
 
 What happens with $\lambda$ in this encoding? It becomes a function from $\N$ to $\N$ and in place of vectors we get integers that perhaps follow some interesting law.
  
 Going back to the set $\{0,1,\ldots, 2^{n}-1\}$,  label the upper half of the points by $x_1=2^{n-1}, x_2= 2^{n-1}+1,\ldots, x_{2^{n-1}}=2^n-1$ (all binary).
 Take as an example $n=4$, then $x_1=1000$, $x_2=1001$,..., $x_8=1111$. Using the definition \ref{deflambda} we can easily calculate the concurrence vectors associated to these numbers: $\lambda(1000)=000111,\lambda(1001)=001100,\ldots,\lambda(1111)=111111$.  As mentioned previously, these points (i.e. the associated 0/1 vectors) are vertices of the polytope {\bf 1}-CUT($n$). Obviously, out of $2^{n-1}$ vertices, $v_1,\ldots, v_{2^{n-1}}$ we can directly identify two vertices of the polytope:  $v_1=\lambda(x_1)=2^{n-1}-1$ and $v_{2^{n-1}}=\lambda(x_{2^{n-1}})=2^{n \choose 2}-1$. Actually, for every $k$ evaluating $\lambda(x_k)$ is straightforward but it is interesting to see is there a more general law between these integers.
 
 \begin{remark}
 Note that $\lambda(0\vdots x)= \bar x \vdots \lambda(x)$ and $\lambda(1\vdots x)=  x \vdots \lambda(x)$.
 \end{remark}
 For example $\lambda(0111)=000111=\overline{111} \vdots \lambda(111)$.

 \begin{proposition}
  For $x,y$ written using same number of bits and with 1 as a leading digit, If $x < y$ then $\lambda(x) < \lambda(y).$ 
 \end{proposition}
 \begin{proof}
 Suppose $x$ and $x+1$ can be written using same number of bits $n$ and have 1 as a leading digit. Then they can be written as $x=2^{n-1}+z$ and $x+1=2^{n-1}+z+1$  Then, $\lambda(x+1)= \lambda(1 \vdots z+1)= z+ 1 \vdots \lambda(z+1) > z \vdots 11...1 \geq z \vdots \lambda(z)= \lambda(x)$. The strict inequality here is due to the fact that 1 at any position to the left from the $\vdots$ has more weight than all ones at the right side. \end{proof}
 
 This proposition tells us that the sequence $v_1,\ldots, v_{2^{n-1}}$ is increasing. Moreover, with appropriate scaling  they are approximately on the line $y = x -1/2$. This is the statement of the next theorem.
 
 \begin{theorem} \label{thm}
 For $k=1,\ldots,2^{n-1}$ and $v_k= \lambda(x_k)$
 \begin{center}
 $|\frac{v_k} {2^{{n-1} \choose 2} }-(k -1/2)| < 1/2.$
 \end{center}
 That is, $v_k/2^{{n-1} \choose 2}$ are approximately  on the line $y=k-1/2$ with residuals being at most 1/2.
 \end{theorem} 
 
 \begin{proof}
 %The proof is done by induction. Suppose the inequality is true for some $n \geq 3$. Then, for $n+1$
 \begin{eqnarray*}
 \lambda(x_{k})= \lambda( 1\vdots (k-1)^b_{[n-1]})= (k-1)^b_{[n-1]} \vdots \lambda((k-1)^b_{[n-1]})\\
 =(k-1) 2^{ n -1 \choose 2}+ \lambda((k-1)^b_{[n-1]})
 \end{eqnarray*}
 Note that $\lambda((k-1)^b_{[n-1]})$ has $n -1  \choose 2$ bits and therefore $\lambda((k-1)^b_{[n-1]})  < 2^{n -1\choose 2}$. Then 
 \begin{eqnarray*}
 \frac{\lambda(x_{k})}{2^{n-1 \choose 2}} - (k- 1/2) = -1/2 +  \lambda((k-1)^b_{[n-1]})/2^{n-1 \choose 2}
 \end{eqnarray*}
 which finishes the proof.
 \end{proof}
 %Treba da specificiramo tu aproksimaciju....
 
 In the following figures $v_k/2^{n-1 \choose 2}$ is plotted versus $k=1,\ldots, 2^{n-1}$, for $n=8$ and $n=12$.  The fitted line $y= x - 0.5$ is obtained using linear regression (MATLAB). For $n=12$ the residuals plot is also shawn, and it can be clearly seen that residuals, in absolute value, are bounded by 1/2 as stated in the above theorem.

\begin{center} 
\includegraphics[width=6in]{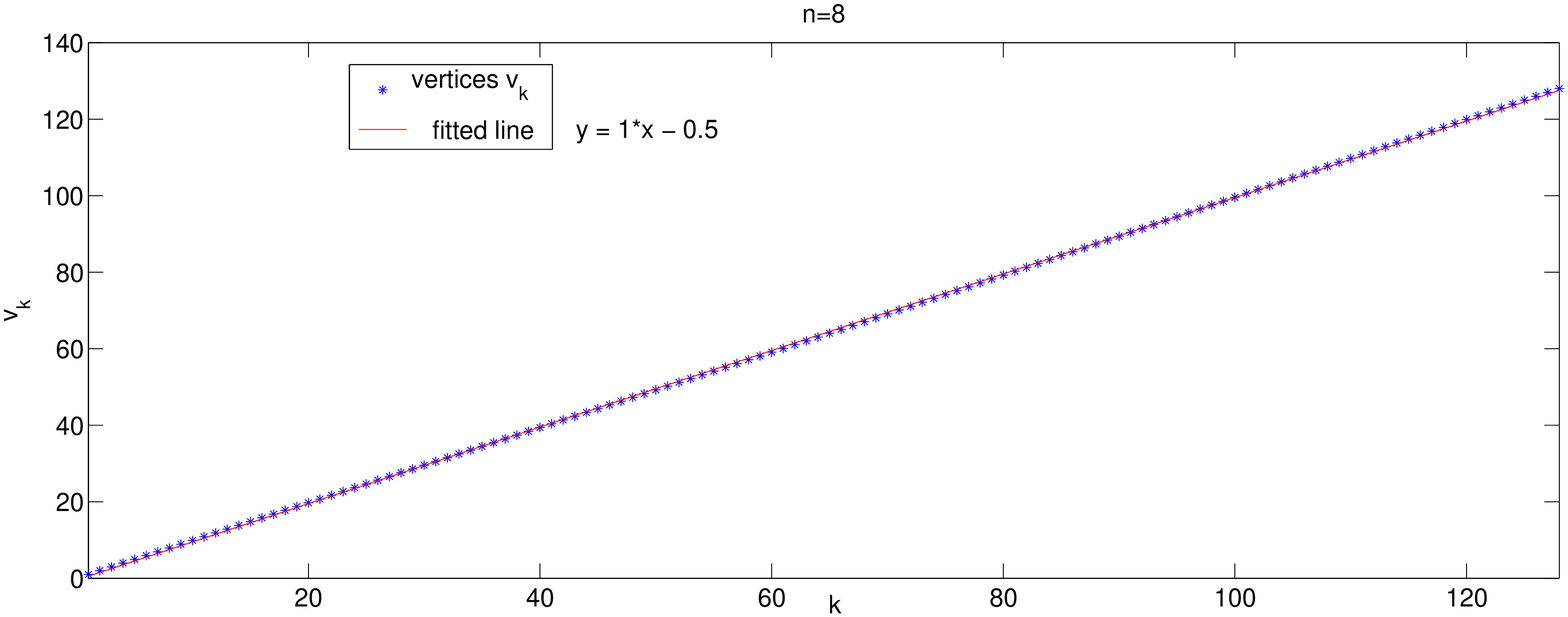} \label{dek8print}  \\
\includegraphics[width=6in]{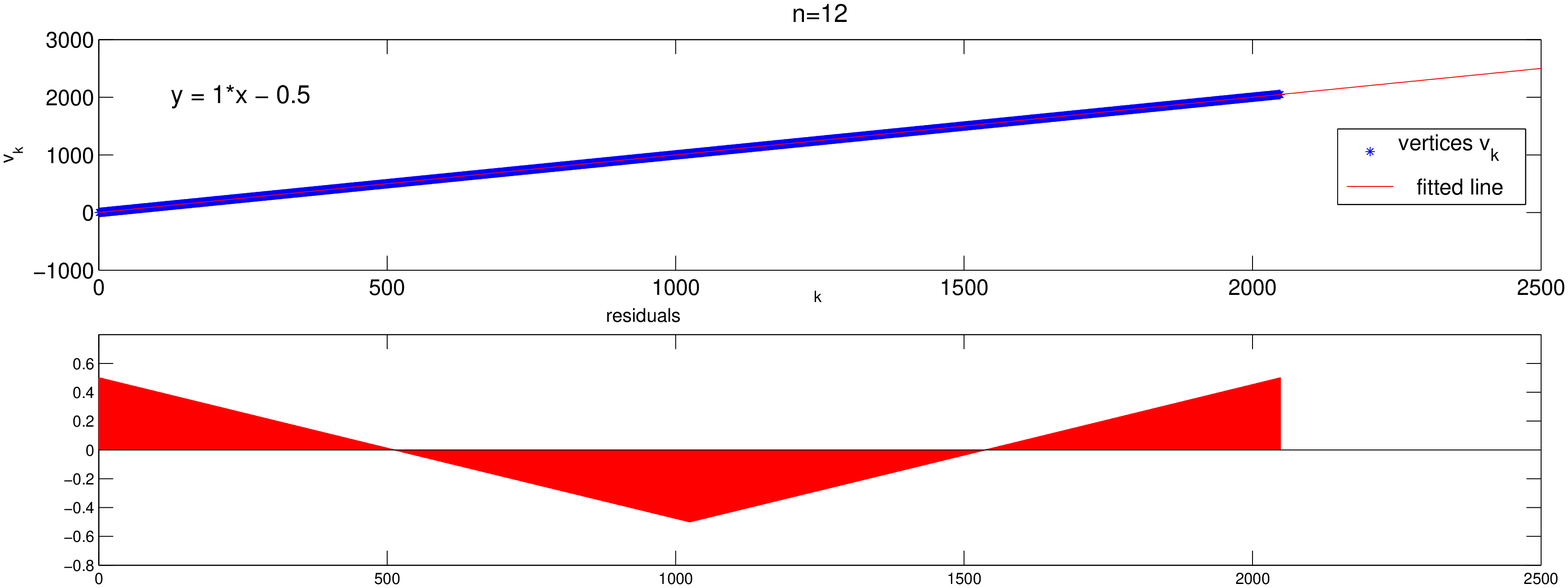} \label{dek12print} 
\end{center}

Note that the above analysis refers to the vertices of a polytope {\bf 1} - CUT($n$). If we denote by $c_k$ the appropriate encoded vertices of CUT($n$), then $c_k= 2^{n \choose 2} - 1 - v_k$.  A corollary of the Theorem \ref{thm} then says that $c_k/2^{{n-1} \choose 2}$ are approximately  on the line $y=-k+2^{n-1}+1/2$.

\bibliographystyle{endm}  
\bibliography{Berncut}

\begin{thebibliography}{1}
\expandafter\ifx\csname url\endcsname\relax
  \def\url#1{\texttt{#1}}\fi
\expandafter\ifx\csname urlprefix\endcsname\relax\def\urlprefix{URL }\fi
\newcommand{\enquote}[1]{``#1''}

\bibitem{dezalaurent1997}
Deza, M.~M. and M.~Laurent, \enquote{Geometry of Cuts and Metrics,}  Algorithms
  and Combinatorics  \textbf{15}, Springer-Verlag, 1997.

\bibitem{huberm2015}
Huber, M. and N.~Mari\'c, \emph{Simulation of multivariate distributions with
  fixed marginals and correlations}, J. Appl. Probab. \textbf{52} (2015),
  pp.~602--608, {a}rXiv:1311.2002.

\bibitem{huberm2017}
Huber, M. and N.~Mari\'c, \emph{Bernoulli correlations and cut polytopes},
  arXiv preprint arXiv:1706.06182  (2017).

\bibitem{tropp2018}
Tropp, J.~A., \emph{Simplicial faces of the set of correlation matrices},
  Discrete \& Computational Geometry \textbf{60} (2018), pp.~512--529.

\bibitem{ziegler2000}
Ziegler, G.~M., \emph{Lectures on 0/1 polytopes}, in: \emph{Polytopes -
  combinatorics and computation}, Birkh\"auser Basel, 2000 pp. 1--41.

\end{thebibliography}

\end{document}